\documentclass[11pt]{amsart} 

\usepackage{cite}
\usepackage[utf8]{inputenc}
\usepackage[T1]{fontenc} 
\usepackage[english]{babel} 
\usepackage{amsmath,amsfonts,amsthm, mathrsfs} 
\usepackage{tikz}
\usepackage{genyoungtabtikz}
\usepackage{graphics}
\usepackage{hyperref}

\theoremstyle{plain}
\newtheorem{theorem}{Theorem}
\newtheorem{corollary}{Corollary}
\newtheorem{lemma}{Lemma}
\newtheorem{proposition}{Proposition}

\theoremstyle{definition}

\DeclareMathOperator{\U}{U}
\DeclareMathOperator{\SU}{SU}

\DeclareMathOperator{\str}{str}

\DeclareMathOperator{\tr}{tr}

\title[Power series identity]{	
A power series identity and Bessel-type integrals over unitary groups
}

\author{Jimmy He} 
\thanks{This research was supported in part by NSERC}
\thanks{The author wishes to thank Dan Bump and Persi Diaconis for helpful conversations}
\address{Department of Mathematics, Stanford University, Stanford, CA  94305}
\email{jimmyhe@stanford.edu}

\begin{document}
\begin{abstract}
In 2008, Lehner, Wettig, Guhr and Wei conjectured a power series identity and showed that it implied a determinantal formula for a Bessel-type integral over the unitary supergroup. The integral is the supersymmetric extension of Bessel-type integrals over the unitary group appearing as partition functions in quantum chromodynamics. The identity is proved by interpreting both sides as the same unitary integral, which can be computed using the Cartan decomposition. An equivalent identity of Schur functions is also given and interpreted probabilistically.
\end{abstract}
\maketitle
\section{Introduction}
In the course of computing certain Bessel-like integrals over the unitary supergroup, Lehner, Wettig, Guhr and Wei conjectured a power series identity which would imply determinantal formulas for those integrals \cite{LWGW08}. The integrals are the supersymmetric analogue of Bessel-type integrals over the classical groups and symmetric spaces. The classical integrals were introduced to the mathematical literature in \cite{B52} and developed in a systematic manner in \cite{H55}. The integrals over the unitary group occur in computing partition functions in quantum chromodynamics \cite{LS92,JSV96} and lattice gauge theory \cite{BG80,GW80,AMM09}. They also appear in statistics in the study of non-central Wishart distributions \cite{H72}.

Use boldface to denote a vector. Given a vector $\mathbf{k}=(k_1,\dotsc,k_n)$, write $\mathbf{z}^\mathbf{k}=z_1^{k_1}\dotsm z_n^{k_n}$. Let $\Delta(\mathbf{z})=\prod_{i<j} (z_i-z_j)$ denote the Vandermonde determinant, and let $C_n=\prod _{i=0}^{n-1}i!$.

The main result proved is the following power series identity.
\begin{theorem}
\label{thm: power series identity}
Let $s\in\mathbf{C}$, and let $z_i\in \mathbf{C}$ for $i=1,\dotsc,m+n$. Then
\begin{equation}
\begin{split}
\label{eq: identity}
&\sum _{\mathbf{k}}\frac{\Delta(\mathbf{k})}{\prod _ik_i!\Gamma(k_i-s+1)}\mathbf{z}^\mathbf{k}
\\=&\sum _{\mathbf{k}}\frac{\Delta(\mathbf{k}_1)\Delta(\mathbf{k}_2)}{\prod _ik_i!\Gamma(k_i-s+1)}\prod_{\substack{{i\leq m}\\{j\geq m+1}}} \frac{z_i-z_j}{k_i+k_j-s+1}\mathbf{z}^\mathbf{k}
\end{split}
\end{equation}
where $\mathbf{k}=(k_1,\dotsc,k_{m+n})$, $\mathbf{k}_1=(k_1,\dotsc,k_m)$ and $\mathbf{k}_2=(k_{m+1},\dotsc,k_{m+n})$. 
\end{theorem}

The proof of Theorem \ref{thm: power series identity} involves interpreting both sides as the same Bessel-type integral over the unitary group. Then the Cartan decomposition of $\U(m+n)$ with respect to $\U(m)\times \U(n)$ and previously known integrals are used to establish the identity for integer $s$. Finally, the result for arbitrary $s$ follows because the coefficients satisfy polynomial relations.

Specializing $s=0$ proves Conjecture D.1 in \cite{LWGW08} and establishes a determinantal formula for Bessel-like integrals over the unitary supergroup.
\begin{theorem}
\label{thm: superintegral formula}
Let $\beta\in \mathbf{C}$ and let $A,B$ be $(m+n)\times(m+n)$ supermatrices. Denote the eigenvalues of $AB$ by $\lambda_i^2$. Then
\begin{equation*}
\begin{split}
&\int _{\U(m|n)}\exp(\beta\str(AX+BX^{-1}))dX
\\=&C_mC_n\beta^{\frac{(m+n)-(m-n)^2}{2}}\frac{\det\left(\lambda_j^{m+n-i}I_{m+n-i}(2\beta\lambda_j)\right)}{\Delta(\lambda_1^2,\dotsc,\lambda_m^2)\Delta(\lambda_{m+1}^2,\dotsc,\lambda_{m+n}^2)}.
\end{split}
\end{equation*}
\end{theorem}
Here, $\U(m|n)$ denotes the unitary supergroup and $I_k$ denotes the modified Bessel function of the first kind of order $k$. The integral is with respect to both commuting and anti-commuting coordinates, and $\str$ denotes the supertrace. Note that there is no canonical scaling for the integral over $\U(m|n)$ and so a choice is made so that when $m=0$ or $n=0$ the result for the unitary group is recovered (see Proposition \ref{prop: G integral}). 

Theorem \ref{thm: superintegral formula} was established in \cite{LWGW08} conditional upon Theorem \ref{thm: power series identity}, and the reader is referred there for the proof of this, which involves expanding the integrand into supercharacters and utilizing orthogonality relations for the matrix coefficients. It is the supersymmetric analogue of similar integrals over the unitary group computed in \cite{SW03} with a similar proof.

The proof of Theorem \ref{thm: power series identity} does not rely on Lie supergroups or superanalysis; see \cite{B87} for general background on superanalysis and integration over Lie supergroups and \cite{LWGW08} and the references therein for the necessary definitions and background to Theorem \ref{thm: superintegral formula}.

Since both sides of \eqref{eq: identity} are anti-symmetric in $\mathbf{z}$, dividing by $\Delta(\mathbf{z})$ gives an identity in terms of Schur functions. A probabilistic interpretation of the identity in terms of random partitions is also given.

\section{The Cartan Decomposition}
\label{sec: Cartan Decomposition}
This section fixes coordinates and notation and explains the Cartan decomposition of $\U(m+n)$. See \cite{H79} or \cite{B04} for further information about the Cartan decomposition and real semisimple Lie groups.

Consider the Lie group $G=\U(m+n)$, with Lie algebra $\mathfrak{g}=\mathfrak{u}(m+n)$ given by skew-Hermitian matrices. Assume from now on that $m\leq n$. There is a Lie subgroup $K=\U(m)\times \U(n)$ embedded block diagonally with Lie algebra $\mathfrak{k}$. For this choice of $G$ and $K$, $(G,K)$ forms a symmetric pair, and so gives a Cartan decomposition $G=KAK$ and $\mathfrak{g}=\mathfrak{k}\oplus i\mathfrak{p}$, where $A=exp(i\mathfrak{a})$ with $i\mathfrak{a}\subseteq i\mathfrak{p}$ a maximal Abelian subalgebra. In particular, $\mathfrak{k}$ is embedded block-diagonally and $i\mathfrak{p}$ consists of all matrices of the form
\begin{equation*}
\left(\begin{array}{cc}
0&X\\
-X^*&0\\
\end{array}\right).
\end{equation*}
Pick the maximal Abelian subalgebra to be the set of matrices of the form
\begin{equation*}
H=\left(\begin{array}{ccc}
0&i\Theta&0\\
i\Theta&0&0\\
0&0&0\\
\end{array}\right)
\end{equation*}
for $\Theta$ an $m\times m$ diagonal matrix with entries $\theta_i$, $\theta_i$ real.

Consider the roots $\Sigma$ of the complexified Lie algebra $\mathfrak{g}_\mathbf{C}$ with respect to a Cartan subalgebra $\mathfrak{h}$ containing $\mathfrak{a}$. Call such roots which are not identically $0$ on $\mathfrak{a}$ \emph{restricted roots}, and denote the set of all restricted roots by $\Phi$. Let $\Phi^+$ denote the set of positive restricted roots (according to some choice of base for the root system).

The collection of all restricted roots forms a root system, although it may not be reduced. Moreover, the corresponding root spaces \begin{equation*}
\mathfrak{g}_\lambda=\{X\in \mathfrak{g}_\mathbf{C}\mid [H,X]=\lambda(H)X\text{ for all }H\in\mathfrak{a}\}
\end{equation*}
may fail to be one-dimensional. Call the dimension of the root space the multiplicity, and denote the multiplicity of $\lambda$ with $m_\lambda$. This is also equal to the number of roots in $\Sigma$ which restrict to $\lambda$ on $\mathfrak{a}$.

A \emph{Weyl chamber} of $\mathfrak{a}$ is a connected component of $\mathfrak{a}$ with the hyperplanes cut out by the restricted roots $\lambda\in \Phi$ removed. If $\Delta$ is a base of simple roots of $\Phi$, then the fundamental Weyl chamber is given by
\begin{equation*}
\mathfrak{a}^+=\{H\in \mathfrak{a}\mid \alpha(H)>0\text{ for all }\alpha\in \Delta\}.
\end{equation*}
A \emph{Weyl alcove} is a connected component of $\mathfrak{a}$ with the affine hyperplanes cut out by $\lambda(H)=n\pi$ for $n\in\mathbf{Z}$ removed. Denote by $\mathfrak{a}_0^+$ the Weyl alcove contained in $\mathfrak{a}^+$ and whose closure contains $0$.

In the case of interest with $G=\U(m+n)$ and $K=\U(m)\times \U(n)$, the restricted root system is of type $BC_n$ for $n>m$ and type $C_n$ for $n=m$, with positive roots $\varepsilon_i\pm \varepsilon_j$, $i>j$, $\varepsilon_i$ (only if $n>m$) and $2\varepsilon_i$, where $\varepsilon_i(H)=i\theta_i$. The corresponding simple roots are given by $\varepsilon_{i+1}-\varepsilon_{i}$ for $1\leq i\leq n-1$ and the root $\varepsilon_1$. With this choice,
\begin{equation*}
\mathfrak{a}_0^+=\left\{H\in\mathfrak{a}\mid 0<\theta_1<\dotso<\theta_n<\frac{\pi}{2}\right\}.
\end{equation*}

The following integration formula for integrals over a compact semisimple Lie group with respect to the Cartan decomposition is essential in the proof of Theorem \ref{thm: power series identity}. The result is due to Harish-Chandra \cite[Lemma 22]{HC56}, but see also \cite[Ch. 1, Theorem 5.10]{H84}.

\begin{proposition}[{\cite[Lemma 22]{HC56}}]
\label{prop: cartan decomp integral}
Let $G$ be a compact semisimple real Lie group, and let $K$ be a compact subgroup such that $(G,K)$ is a symmetric pair. Let $\mathfrak{g}=\mathfrak{k}\oplus i\mathfrak{p}$ be the Cartan decomposition and let $i\mathfrak{a}\subseteq i\mathfrak{p}$ be a maximal Abelian subalgebra. Then the measures on $G$, $K$ and $i\mathfrak{a}_0^+$ (Haar measure and Lebesgue measure respectively) may be normalized such that
\begin{equation*}
\int _{G}f(x)dx=\int _{K\times i\mathfrak{a}_0^+\times K}f(k_1e^Hk_2)J(H)dk_1dk_2dH
\end{equation*}
where
\begin{equation*}
J(H)=\prod_{\lambda\in \Phi^+} |\sin(-i\lambda(H))|^{m_\lambda}.
\end{equation*}
\end{proposition}
Note that although $\U(m+n)$ is not semisimple, it is a semidirect product of $\SU(m+n)$ with $\U(1)$ and so it is easy to see that the formula still holds. In the case of interest with $G=\U(m+n)$ and $K=\U(m)\times \U(n)$, Table \ref{tab: root multiplicity} gives the multiplicities of the restricted roots \cite[Table VI]{H79}. Then
\begin{equation*}
J(H)=\prod _{1\leq i<j\leq m}\left(\sin(\theta_i-\theta_j)^2\sin(\theta_i+\theta_j)^2\right)\prod _{i=1}^m\left(\sin(\theta_i)^{2(n-m)}\sin(2\theta_i)\right).
\end{equation*}

\begin{table}
\centering
\caption{Multiplicities of restricted roots for $\U(m+n)/\U(m)\times \U(n)$}
\label{tab: root multiplicity}
\begin{tabular}{|c|c|}
\hline
$\lambda$ & $m_\lambda$\\\hline
$\varepsilon_i-\varepsilon_j$ & $2$\\
$\varepsilon_i+\varepsilon_j$ & $2$\\
$\varepsilon_i$ & $2(n-m)$\\
$2\varepsilon_i$ & $1$\\\hline
\end{tabular}
\end{table}

\section{Unitary Integrals}
\label{sec: unitary integrals}
The following power series expansions for unitary integrals were computed in \cite{SW03} using character expansion methods, which expand the integrand into characters and then use orthogonality relations to simplify the result.

\begin{proposition}[{\cite[Eq. 14]{SW03}}]
\label{prop: G integral}
Let $A,B$ be $n\times n$ matrices with complex entries, with $AB$ diagonalizable, and let $s$ be a non-negative integer. Denote the eigenvalues of $AB$ by $z_i$. Then
\begin{equation*}
\int _{\U(n)}\det(AU)^s\exp(\tr(AU+BU^{-1}))dU=\frac{C_{n}}{\Delta(\mathbf{z})}\sum _{\mathbf{k}}\frac{\Delta(\mathbf{k})}{\prod _{i}k_i!(k_i-s)!}\mathbf{z}^\mathbf{k}.
\end{equation*}
\end{proposition}
Note that Proposition \ref{prop: G integral} differs slightly from what is in \cite{SW03} because the identity (see e.g. \cite[Thm. C.3]{LWGW08})
\begin{equation*}
\det\left(\frac{1}{k_j-n-s-i}\right)=\frac{\Delta(\mathbf{k})}{\prod k_i!}
\end{equation*}
is used.

\begin{proposition}[{\cite[Eq. 23]{SW03}}]
\label{prop: K integral}
Let $A,B,C,D$ be $n\times n$ matrices with complex entries, with $AD$ and $BC$ diagonalizable, and let $s$ be a non-negative integer. Denote the eigenvalues of $AD$ by $x_i$ and those of $BC$ by $y_i$. Then
\begin{equation*}
\begin{split}
&\int _{\U(n)\times \U(n)}\det(UAVB)^s\exp(\tr(UAVB+U^{-1}CV^{-1}D))dU dV
\\=&\frac{C_n^2}{\Delta(\mathbf{x})\Delta(\mathbf{y})}\sum_{\mathbf{k}}\frac{1}{\prod k_i!(k_i-s)!}\det(x_i^{k_j})\mathbf{y}^\mathbf{k}.
\end{split}
\end{equation*}
\end{proposition}

In both Proposition \ref{prop: G integral} and Proposition \ref{prop: K integral}, the formulas still make sense if the eigenvalues of $AB$ or $AD$ and $BC$ are not distinct. The poles coming from the Vandermonde determinant $\Delta(\mathbf{z})$ are canceled because the sums are anti-symmetric in $\mathbf{z}$. In particular, a limit may be taken to compute the integrals explicitly.

\section{Proof of Theorem \ref{thm: power series identity}}
\label{sec: proof}
In this section, Theorem \ref{thm: power series identity} is proved. This is done by expressing both sides of \eqref{eq: identity} as the same integral over the unitary group, using the Cartan decomposition and the unitary integrals in Section \ref{sec: unitary integrals}.

First, a simple extension of an integral version of the Cauchy-Binet formula, found in \cite{A83}, is developed, which may be of independent interest. If $A,B$ are $k\times m$ and $k\times n$ matrices, then write $(A\mid B)$ to denote the matrix whose first $m$ columns are given by $A$ and whose last $n$ columns are given by $B$.

\begin{lemma}
\label{lemma: cauchy-binet}
Let $m\leq n$ and let $f_i$ and $g_j$, for $1\leq i\leq m$ and $1\leq j\leq n$, be functions on some measure space $\Omega$ such that $f_ig_j$ is integrable for all $i,j$ and $C$ some $n\times (n-m)$ matrix. Then
\begin{equation*}
\int _{\Omega^m}\det(f_i(x_j))\det(g_i(x_j)\mid C)dx_1\dotsm dx_m=m!\det\left(\int _{\Omega}f_i(x)g_j(x)\;\middle|\; C \right).
\end{equation*}
\end{lemma}
\begin{proof}
First use the Laplace expansion along the last $n-m$ columns of $\det(g_i(x_j)|C)$ to obtain
\begin{equation*}
\det(g_i(x_j)\mid C)=\sum _{|S|=m}\varepsilon(S)\det((g_i(x_j))_S)\det(C_{S^c})
\end{equation*}
where for a matrix $X$, $X_S$ denotes the submatrix given by taking the rows in $S$ and $\varepsilon(S)$ is a sign determined by $S$. Then an integration formula of Andr\'eief (see \cite{A83}) gives
\begin{equation*}
\int _{\Omega^m}\det(f_i(x_j))\det((g_i(x_j))_S)dx_1\dotsm dx_m=m!\det\left(\int _{\Omega}f_i(x)g_j(x)dx\right)_S
\end{equation*}
and this gives the result after another application of the Laplace expansion.
\end{proof}

Theorem \ref{thm: power series identity} is proved for integers by first establishing that it holds up to a scalar, and then showing that the scalar must be $1$.

\begin{lemma}
\label{lemma: main lemma}
Let $s$ be a non-negative integer and let $z_i\in\mathbf{C}$ for $i=1,\dotsc,m+n$. Then there exists a constant $C$ such that
\begin{equation}
\label{eq: main lemma}
C\sum _{\mathbf{k}}\frac{\Delta(\mathbf{k})}{\prod _ik_i!(k_i-s)!}\mathbf{z}^\mathbf{k}=\sum _{\mathbf{k}}\frac{\Delta(\mathbf{k}_1)\Delta(\mathbf{k}_2)}{\prod _ik_i!(k_i-s)!}\prod_{\substack{{i\leq m}\\{j\geq m+1}}} \frac{z_i-z_j}{k_i+k_j-s+1}\mathbf{z}^\mathbf{k}
\end{equation}
where $\mathbf{k}=(k_1,\dotsc,k_{m+n})$, $\mathbf{k}_1=(k_1,\dotsc,k_m)$ and $\mathbf{k}_2=(k_{m+1},\dotsc,k_{m+n})$.
\end{lemma}

\begin{proof}
Let $Z$ be a diagonal matrix with entries $z_i$. The integral
\begin{equation}
\label{exp: integral over U(m_n)}
\int _{\U(m+n)}\det(ZX)^s\exp(\tr(ZX+X^{-1}))dX
\end{equation}
will be evaluated in two ways, which will give the two sides of the desired identity. First, by Proposition \ref{prop: G integral}, 
\begin{equation*}
\int _{\U(m+n)}\det(ZX)^s\exp(\tr(ZX+X^{-1}))dX=\frac{C_{m+n}}{\Delta(\mathbf{z})}\sum _{\mathbf{k}}\frac{\Delta(\mathbf{k})}{\prod k_i!(k_i-s)!}\mathbf{z}^\mathbf{k}
\end{equation*}
where the sum is over $k_i\geq s$.

Next, the same integral is evaluated using the Cartan decomposition. By Proposition \ref{prop: cartan decomp integral} the integral in \eqref{exp: integral over U(m_n)} is equal to
\begin{equation}
\label{exp: cartan integral}
\int _{K\times i\mathfrak{a}^+_0\times K}\det(ZUe^HV)\exp(\tr(ZUe^HV+V^{-1}e^HU^{-1}))J(H)dUdHdV.
\end{equation}
Now
\begin{equation*}
e^H=\left(\begin{array}{ccc}
\cos(\Theta)&i\sin(\Theta)&0\\
i\sin(\Theta)&\cos(\Theta)&0\\
0&0&I\\
\end{array}\right)
\end{equation*}
and so if $U_i$, $V_i$ and $Z_i$, for $i=1,2$, are the two blocks of $U$, $V$ and $Z$ respectively, since $U$, $V$ and $Z$ are all block-diagonal,
\begin{equation*}
\begin{split}
&\tr(ZUe^HV+V^{-1}e^HU^{-1})
\\=&\tr(Z_1U_1\cos(\Theta)V_1+V_1^{-1}\cos(\Theta)U_1^{-1})+\tr(Z_2U_2\cos(\widetilde{\Theta})V_2+V_2^{-1}\cos(\widetilde{\Theta})U_2^{-1})
\end{split}
\end{equation*}
where
\begin{equation*}
\widetilde{\Theta}=\left(\begin{array}{cc}
\Theta&0\\
0&I\\
\end{array}\right).
\end{equation*}
Then factor the integrand of \eqref{exp: cartan integral} as
\begin{equation*}
\begin{split}
&\det(U_1V_1Z_1)^s\exp(\tr(U_1\cos(\Theta)V_1Z_1+U_1^{-1}V_1^{-1}\cos(\Theta)))
\\&\qquad\times \det(U_2V_2Z_2)^s\exp(\tr(U_2\cos(\widetilde{\Theta})V_2Z_2+U_2^{-1}V_2^{-1}\cos(\widetilde{\Theta})))
\end{split}
\end{equation*}
and apply Proposition \ref{prop: K integral} to compute the integral \eqref{exp: cartan integral} over $\U(m)\times \U(m)$ and $\U(n)\times \U(n)$. The first factor poses no issues and gives
\begin{equation}
\label{exp: factor 1}
\frac{C_m^2}{\Delta(\mathbf{z}_1)}\sum _{\mathbf{k}}\frac{1}{\prod_{i=1}^mk_i!(k_i-s)!}\frac{\det(\cos(\theta_i)^{2k_j-s})}{\Delta(\cos(\theta_1)^2,\dotsc,\cos(\theta_m)^2)}\mathbf{z}_1^{\mathbf{k}}
\end{equation}

where $\mathbf{z}_1=(z_1,\dotsc,z_m)$. For the second factor, note that $\cos(\widetilde{\Theta})$ has repeated eigenvalues if $n-m\geq 2$, so a limit must be taken.

Write 
\begin{equation*}
\widetilde{\Theta}=\lim _{x_i\rightarrow 1}\left(\begin{array}
{cc}
\Theta&0\\
0&X\\
\end{array}\right)
\end{equation*}
where $X$ is diagonal with entries $x_i$. Then applying Proposition \ref{prop: K integral} gives 
\begin{equation*}
\lim_{x_i\rightarrow 1}\frac{C_n^2}{\Delta(\mathbf{z}_2)}\sum _{\mathbf{l}}\frac{1}{\prod_{i=1}^ml_i!(l_i-s)!}\frac{\det\left(\cos(\theta_i)^{2l_j-s}\;\middle|\; x_i^{l_j}\right)}{\Delta(\cos(\theta_1)^2,\dotsc,\cos(\theta_m)^2,x_1,\dotsc,x_{n-m})}\mathbf{z}_2^{\mathbf{k}}
\end{equation*}
Now write
\begin{equation*}
\begin{split}
&\Delta\left(\cos(\theta_1)^2,\dotsc,\cos(\theta_m)^2,x_1,\dotsc,x_{n-m}\right)
\\=&\Delta\left(\cos(\theta_1)^2,\dotsc,\cos(\theta_m)^2\right)\Delta(x_1,\dotsc,x_{n-m})\prod _{i=1}^m\prod_{j=1}^{n-m} \left(\cos(\theta_i)^2-x_j\right)
\end{split}
\end{equation*}
and compute
\begin{equation*}
\lim _{x_i\rightarrow 1}\frac{\det\left(\cos(\theta_i)^{2l_j-s}\;\middle|\; x_i^{l_j}\right)}{\Delta(x_1,\dotsc,x_{n-m})}
\end{equation*}
by taking the derivative with respect to $x_i$ $i-1$ times and setting $x_i=1$. Because
\begin{equation*}
\Delta(x_1,\dotsc,x_{n-m})=\det\left(\begin{array}{cccc}
x_1^{n-m-1}&\dots&x_{n-m}^{n-m-1}\\
x_1^{n-m-2}&\dots&x_{n-m}^{n-m-2}\\
\vdots&\ddots&\vdots\\
1&\dots&1\\
\end{array}\right)
\end{equation*}
the denominator will become $(-1)^{(n-m)(n-m-1)/2}C_{n-m}$, and the numerator will become
\begin{equation*}
\det\left(\cos(\theta_i)^{2l_j-s}\;\middle|\; (l_j)_{i-1}\right)=\det\left(\cos(\theta_i)^{2l_j-s}\;\middle|\; l_j^{i-1}\right)
\end{equation*}
where $(x)_n=x(x-1)\dotsm(x-n+1)$ denotes the falling factorial. Thus, the integral over $\U(n)\times \U(n)$ gives
\begin{equation}
\label{exp: factor 2}
\frac{C}{\Delta(\mathbf{z}_2)}\sum _{\mathbf{l}}\frac{1}{\prod_{i=1}^ml_i!(l_i-s)!}\frac{\det\left(\cos(\theta_i)^{2l_j-s}\;\middle|\; l_j^{i-1}\right)}{\Delta(\cos(\theta_1)^2,\dotsc,\cos(\theta_m)^2)\prod _{i}(1-\cos(\theta_i)^2)^{n-m}}\mathbf{z}_2^{\mathbf{k}}
\end{equation}
for some constant $C$ (which will be allowed to change from line to line).

Now rewrite $J(H)$ giving
\begin{equation*}
\begin{split}
J(H)&=\prod_{i<j} \sin(\theta_i-\theta_j)^2\sin(\theta_i+\theta_j)^2\prod _i\sin(\theta_i)^{2(n-m)}\sin(2\theta_i)
\\&=\Delta(\cos(\theta_1)^2,\dotsc,\cos(\theta_m)^2)^2\prod _i(1-\cos(\theta_i)^2)^{n-m}
\end{split}
\end{equation*}
which cancels out the corresponding factors in the denominators of the two factors, \eqref{exp: factor 1} and \eqref{exp: factor 2}. Thus,
\begin{equation*}
\begin{split}
&\int _{\U(m+n)}\det(ZX)^s\exp(\tr(ZX+X^{-1}))dX
\\=&\frac{C}{\Delta(\mathbf{z}_1)\Delta(\mathbf{z}_2)}\sum _{\mathbf{k},\mathbf{l}}\left(\frac{\mathbf{z}_1^\mathbf{k}\mathbf{z}_2^\mathbf{l}}{\prod_i k_i!(k_i-s)!\prod _j l_j!(l_j-s)!}\right.
\\&\qquad\left.\times\int _{i\mathfrak{a}_0^+}\det\left(\cos(\theta_i)^{2k_j-s+1}\sin(\theta_i)\right)\det\left(\cos(\theta_i)^{2l_j-s}\;\middle|\; l_j^{i-1}\right)d\theta_1\dotsm d\theta_m\right).
\end{split}
\end{equation*}

Finally, note that the integral over the Weyl alcove can be written as an integral over $[0,\pi/2]^m$ at the cost of a factor of $m!$ because the integrand is symmetric with respect to the $\theta_i$. Then apply Lemma \ref{lemma: cauchy-binet} to evaluate the integral as
\begin{equation*}
\begin{split}
&\int _{i\mathfrak{a}_0^+}\det\left(\cos(\theta_i)^{2k_j-s+1}\sin(\theta_i)\right)\det\left(\cos(\theta_i)^{2l_j-s}\;\middle|\; x_i^{l_j}\right)d\theta_1\dotsm d\theta_m
\\=&\det\left(\int_0^{\pi/2}\cos(\theta)^{2k_i+2l_j-2s+1}\sin(\theta)d\theta\;\middle|\;l_j^{i-1}\right)
\\=&2^{-m}\det\left(\frac{1}{k_i+l_j-s+1}\;\middle|\;l_j^{i-1}\right).
\end{split}
\end{equation*}
Then as
\begin{equation*}
\det\left(\frac{1}{k_i+l_j-s+1}\;\middle|\;l_j^{i-1}\right)=(-1)^{m(n-m)+{n-m\choose 2}}\frac{\Delta(\mathbf{k})\Delta(\mathbf{l})}{\prod _{i,j}(k_i+l_j-s+1)},
\end{equation*}
see \cite[Lemma 2]{BF94}, this establishes Lemma \ref{lemma: main lemma} after multiplying both sides by $\Delta(\mathbf{z})$.
\end{proof}

The computation of the constant $C$ in Lemma \ref{lemma: main lemma} can be reduced to the case where $m=n=1$, where it can be explicitly computed. Finally, $s$ can be extended to an arbitrary complex constant by comparing coefficients.

\begin{proof}[Proof of Theorem \ref{thm: power series identity}]
First, notice that the restriction that $s$ be non-negative can be removed by dividing both sides by $\prod _{i=1}^{m+n}z_i^s$, and re-indexing the summation. In particular, the constant does not change.

Now suppose that
\begin{align}
\label{exp: LHS}
&C\sum_{\mathbf{k}} \frac{\Delta(\mathbf{k})}{\prod _{i}k_i!(k_i+s)!}z_1^{k_1}\dotsm z_{m+n}^{k_{m+n}}
\\=&\sum _{\mathbf{k}}\frac{\Delta(\mathbf{k}_1)\Delta(\mathbf{k}_2)}{\prod_{i} k_i!(k_i+s)!}\prod _{\substack{i\leq m\\j\geq m+1}}\frac{z_i-z_j}{k_i+k_j+s+1}z_1^{k_1}\dotsm z_{m+n}^{k_{m+n}}
\label{exp: RHS}
\end{align}
holds for some $m,n$ and $s$ a non-negative integer. Then setting $z_{m+n}=0$ in \eqref{exp: LHS} gives
\begin{equation*}
\frac{C}{s!}\sum_{k_1,\dotsc,k_{m+n-1}} \frac{\Delta(k_1,\dotsc,k_{m+n-1})}{\prod _{i}k_i!(k_i+s+1)!}z_1^{k_1+1}\dotsm z_{m+n}^{k_{m+n-1}+1}
\end{equation*}
after re-indexing. Similarly setting $z_{m+n}=0$ in \eqref{exp: RHS} gives
\begin{equation*}
\begin{split}
&\frac{1}{s!}\sum _{k_1,\dotsc,k_{m+n}}\frac{\Delta(k_1,\dotsc,k_m)\Delta(k_{m+1},\dotsc,k_{m+n-1})}{\prod_{i\leq m} k_i!(k_i+s+1)!}
\\&\qquad\times\prod _{\substack{i\leq m\\j\geq m+1}}\frac{z_i-z_j}{k_i+k_j+s+1+1}z_1^{k_1+1}\dotsm z_{m+n-1}^{k_{m+n-1}+1}.
\end{split}
\end{equation*}
Then dividing through by $\prod z_i$ and multiplying by $s!$ gives the same identity as in \eqref{eq: main lemma} (with the same constant), with parameters $m,n-1,s+1$ instead of $m,n,s$. By setting $z_1=0$, the same identity with parameters $m-1,n,s+1$ is obtained. This reduces to the case of $m=n=1$, with arbitrary integer $s$. Then the following computation
\begin{equation*}
\begin{split}
&\sum _{k_1,k_2}\frac{k_1-k_2}{k_1!k_2!(k_1+s)!(k_2+s)!}z_1^{k_1}z_2^{k_2}
\\=&\sum _{k_1,k_2}\frac{k_1(k_1+s)-k_2(k_2+s)}{k_1!k_2!(k_1+s)!(k_2+s)!}\frac{1}{k_1+k_2+s}z_1^{k_1}z_2^{k_2}
\\=&\sum _{k_1,k_2}\frac{1}{k_1!k_2!(k_1+s)!(k_2+s)!}\frac{z_1-z_2}{k_1+k_2+s+1}z_1^{k_1}z_2^{k_2}
\end{split}
\end{equation*}
for the $m=n=1$ case shows that $C=1$.

To see that the result holds for $s\in\mathbf{C}$, first note that it suffices to prove
\begin{equation}
\label{eq: coeff}
\frac{\Delta(\mathbf{k})}{\prod _i k_i!}=[\mathbf{z}^\mathbf{k}]\sum _{\mathbf{l}}\frac{\Delta(\mathbf{l}_1)\Delta(\mathbf{l}_2)\Gamma(k_i-s+1)}{\prod _il_i!\Gamma(l_i-s+1)}\prod_{\substack{{i\leq m}\\{j\geq m+1}}} \frac{z_i-z_j}{l_i+l_j-s+1}\mathbf{z}^\mathbf{l}
\end{equation}
holds for all $\mathbf{k}$, where $[\mathbf{z}^\mathbf{k}]f(\mathbf{z})$ denotes the coefficient of $\mathbf{z}^\mathbf{k}$ in $f(\mathbf{z})$. But notice that $[\mathbf{z}^\mathbf{k}]\prod (z_i-z_j)\mathbf{z}^{\mathbf{l}}=0$ for all but finitely many terms in the sum in \eqref{eq: coeff}, and $\Gamma(k_i-s+1)/\Gamma(l_i-s+1)$ is a rational function in $s$, and so both sides of \eqref{eq: coeff} are rational functions of $s$. Furthermore, it was already shown that \eqref{eq: coeff} holds for integer $s$, and so it must hold for $s\in\mathbf{C}$.
\end{proof}

\section{Probabilistic Interpretation}
In this section, Theorem \ref{thm: power series identity} is rewritten in terms of Schur functions. The coefficients are then interpreted as probabilities on partitions, and the limiting distribution is found. See \cite{M79} for background on symmetric functions and see \cite{B04} for more on the representation theory of $\SU(n)$.

Let $\lambda=(\lambda_1,\dotsc,\lambda_k)$ denote a partition of length $k$. For $i>k$, let $\lambda_i=0$. Another way of viewing a partition is as a Young diagram, an array of boxes with $\lambda_i$ boxes in the $i$th row, justified left.

Let $k_i=\lambda_i+n-i$ and let
\begin{equation*}
s_\lambda(z_1,\dotsc,z_n)=\frac{\det(z_i^{k_j})}{\Delta(z_1,\dotsc,z_n)}
\end{equation*}
be the Schur functions. These functions are symmetric polynomials in $z_i$ and have applications in representation theory and combinatorics. Use $s_\lambda(1^n)$ to denote $s_\lambda(1,\dotsc,1)$ with $n$ variables.

The \emph{hook length} of $s\in \lambda$ is defined to be $1$ plus the number of boxes below and to the right of $s$. Given a partition $\lambda=(\lambda_1,\dotsc,\lambda_n)$, let $h(\lambda)$ denote the product of the hook lengths of the boxes in $\lambda$. For a partition of length at most $n$, let $k_i=\lambda_i+n-i$. Then
\begin{equation}
\label{eq: hook length}
h(\lambda)=\frac{\prod_{i=1}^nk_i!}{\Delta(k_1,\dotsc,k_n)}.
\end{equation}

Let $(k^l)$ denote the partition of length $l$, with all parts equal to $k$. Finally, given two partitions $\lambda,\mu$, let $\lambda+\mu$ denote $(\lambda_1+\mu_1,\dotsc,\lambda_n+\mu_n)$, $\lambda/\mu$ denote $(\lambda_1-\mu_1,\dotsc,\lambda_n-\mu_n)$ (assuming $\lambda_i\geq \mu_i$ for all $i$), and $\lambda\cup \mu$ denote the partition given by appending $\mu$ to the end of $\lambda$ (assuming $\mu_1\leq \lambda_n$).

\begin{corollary}
\label{cor: schur identity}
Let $s$ be a non-negative integer. Then
\begin{equation}
\label{eq: schur identity}
\sum _{\lambda}a_s(\lambda)\frac{s_\lambda(z_1,\dotsc,z_{m+n})}{s_\lambda(1^{m+n})}=\sum _{\mu,\nu}b_s(\mu,\nu)\frac{s_\mu(z_1,\dotsc,z_m)s_\nu(z_{m+1},\dotsc,z_{m+n})}{s_\mu(1^m)s_\nu(1^n)}
\end{equation}
where
\begin{align*}
a_s(\lambda)&=\frac{1}{h(\lambda)h(\lambda+(s^{m+n}))},
\\b_s(\mu,\nu)&=\frac{h((n^m))}{h(\mu)h(\nu+(s^n))h((\mu+(s^m)+(n^m))\cup\nu')}
\\&=\frac{1}{h(\mu)h(\nu)h(\mu+(s^m))h(\nu+(s^n))}
\\&\qquad\qquad\times\prod _{i=1}^m\prod _{j=1}^n\frac{m+n-i-j+1}{\mu_i+\nu_j+m+n-i-j+s+1}.
\end{align*}
\end{corollary}
\begin{proof}
Divide both sides of \eqref{eq: identity} by $\Delta(\mathbf{z})$. The coefficients of
\begin{equation}
\label{exp: LHS schur}
\frac{1}{\Delta(\mathbf{z})}\sum _{\mathbf{k}}\frac{\Delta(\mathbf{k})}{\prod _ik_i!(k_i+s)!}\mathbf{z}^\mathbf{k}
\end{equation}
are antisymmetric in $\mathbf{z}$, and so sum over $k_1>\dotso>k_{m+n}$ and let $\lambda_i=k_i-m-n+i$. Then \eqref{exp: LHS schur} is equal to
\begin{equation*}
\sum _{\lambda}\frac{\Delta(\mathbf{k})}{\prod _ik_i!(k_i+s)!}\frac{\det(z_i^{k_j})}{\Delta(\mathbf{z})}=\sum _{\lambda}\frac{\Delta(\mathbf{k})}{\prod _ik_i!(k_i+s)!}s_\lambda(z_1,\dotsc,z_{m+n}).
\end{equation*}
Similarly,
\begin{equation*}
\begin{split}
&\frac{1}{\Delta(\mathbf{z})}\sum _{\mathbf{k}}\frac{\Delta(\mathbf{k}_1)\Delta(\mathbf{k}_2)}{\prod _ik_i!(k_i+s)!}\prod_{\substack{{i\leq m}\\{j\geq m+1}}} \frac{z_i-z_j}{k_i+k_j+s+1}\mathbf{z}^\mathbf{k}
\\=&\sum _{\mu,\nu}\frac{\Delta(\mathbf{k}_1)\Delta(\mathbf{k}_2)}{\prod _ik_i!(k_i+s)!}\prod_{\substack{{i\leq m}\\{j\geq m+1}}} \frac{1}{k_i+k_j+s+1}s_\mu(z_1,\dotsc,z_m)s_\nu(z_{m+1},\dotsc,z_{m+n}).
\end{split}
\end{equation*}
By the formula for the hook length given by \eqref{eq: hook length} and the formula $s_\lambda(1^{m+n})=\frac{\Delta(\mathbf{k})}{C_{m+n}}$ (and similarly for $s_\mu$ and $s_\nu$), the identity
\begin{equation*}
\begin{split}
&\sum_{\lambda}\frac{1}{h(\lambda)h(\lambda+(s^{m+n}))}\frac{s_\lambda(z_1,\dotsc,z_{m+n})}{s_\lambda(1^{m+n})}
\\=&\sum _{\mu,\nu}\frac{C_{m+n}}{C_mC_nh(\mu)h(\mu+(s^m))h(\nu)h(\nu+(s^n))}
\\&\qquad\times\prod_{\substack{{i\leq m}\\{j\geq m+1}}} \frac{1}{k_i+k_j+s+1}\frac{s_\mu(z_1,\dotsc,z_{m})s_\nu(z_{m+1},\dotsc,z_{m+n})}{s_\mu(1^m)s_\nu(1^n)}
\end{split}
\end{equation*}
is obtained. Finally, noting that $h((n^m))=\frac{C_{m+n}}{C_mC_n}$ and $\frac{1}{k_i+k_j+s+1}$ is the hook length of the box $(i,j)$ in $(\nu+(s^m)+(n^m))\cup\nu'$ gives the result (see Figure \ref{fig: joined partition} for an example).

\begin{figure}
\caption{The partition $(\mu+(s^m)+(n^m))\cup \nu'$ for $s=2$, $m=2$, $n=3$}
\label{fig: joined partition}
\begin{tikzpicture}
\Yfillopacity{0}
\Ylinecolour{gray}
\Ylinethick{0.05pt}
\tgyoung(0cm,0cm,;;;;;;;,;;;;;;,;;,;;,;,;)
\draw[thick] (0pt,13pt)--(0pt,-65pt)--(13pt,-65pt)--(13pt,-39pt)--(26pt,-39pt)--(26pt,-13pt)--(78pt,-13pt)--(78pt,0pt)--(91pt,0pt)--(91pt,13pt)--(0pt,13pt);
\draw[thick] (0pt,-13pt)--(39pt,-13pt)--(39pt,13pt);
\draw[thick] (65pt,13pt)--(65pt,-13pt);
\draw (0pt,13pt) -- (0pt,-13pt) node [black,midway,xshift=-10pt]{$m$};
\draw (0pt,13pt) -- (39pt,13pt) node [black,midway,yshift=6pt]{$n$};
\draw (39pt,13pt) -- (65pt,13pt) node [black,midway,yshift=6pt]{$s$};
\draw (26pt,-50pt) node{$\nu'$};
\draw (80pt,-20pt) node{$\mu$};
\end{tikzpicture}
\end{figure}
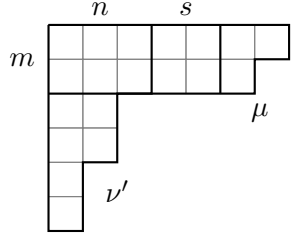
\end{proof}
Note that if $s$ is negative, $\lambda+(s^{m+n})$ may be interpreted as $\lambda/(-s)^{m+n}$, and the sum is over $\lambda$ containing $(-s)^{m+n}$. Also, note that the coefficient
\begin{equation*}
\frac{h((n^m))}{h(\mu)h(\nu+(s^n))h((\mu+(s^m)+(n^m))\cup\nu')}
\end{equation*}
is actually symmetric in $\mu,\nu$. Since the Schur functions are a basis for the ring of symmetric functions, $s_\lambda(z_1,\dotsc,z_{m+n})$ can be expanded in terms of $s_\mu(z_1,\dotsc,z_m)s_\nu(z_{m+1},\dotsc,z_{m+n})$, giving
\begin{equation*}
s_\lambda(z_1,\dotsc,z_{m+n})=\sum _{\mu,\nu}c_{\mu\nu}^\lambda s_\mu(z_1,\dotsc,z_m)s_\nu(z_{m+1},\dotsc,z_{m+n}).
\end{equation*}
The coefficients $c_{\mu\nu}^\lambda$ are known as the Littlewood-Richardson coefficients, and they appear as multiplicities of irreducible representations in tensor products and induced and restricted representations (see \cite[Ch. 38]{B04}). For example, if $V_\lambda$ is the representation of $\SU(m+n)$ with highest weight $\lambda$, then the multiplicity of $V_{\mu}\otimes V_\nu$ in the restriction of $V_\lambda$ to $\SU(m)\times \SU(n)$ is $c_{\mu\nu}^\lambda$. 

Note that in this form, both sides of \eqref{eq: schur identity} are Schur generating functions as defined in \cite{BG16}, and the coefficients may be viewed as (unnormalized) probabilities on the space of partitions. The normalization constant can be computed explicitly in a determinantal form either from the formulas given in \cite{SW03} (although note that due to the repeated eigenvalues, a non-trivial limit must be taken), or directly from the integral
\begin{equation*}
\int \det(X)^s\exp(\tr(X+X^{-1}))dX
\end{equation*}
which can be evaluated using the integral form of the Cauchy-Binet identity due to Andr\'eief \cite{A83} after applying the Weyl integration formula. Either method gives
\begin{equation*}
Z_{m,n,s}=\sum_{l(\lambda)\leq m+n} \frac{1}{h(\lambda)h(\lambda+(s^{m+n}))}=\det(I_{i-j+s}(2))
\end{equation*}
where the matrix is $(m+n)\times (m+n)$.

The identity may then be interpreted as a statement about random partitions (or equivalently random irreducible representations of $\SU(n)$). Pick an irreducible representation of $\SU(m+n)$ according to the measure 
\begin{equation*}
P_{m,n,s}(\lambda)=\frac{1}{Z_{m,n,s}h(\lambda)h(\lambda+(s^{m+n}))}.
\end{equation*}
Then restrict the representation $V_\lambda$ to $\SU(m)\times \SU(n)$, and pick an irreducible representation of $\U(m)\times \U(n)$ appearing in $V_\lambda$ with probability $c_{\mu\nu}^\lambda s_\mu(1^m)s_\nu(1^n)/s_\lambda(1^{m+n})$. This is exactly a random irreducible constituent of $V_\lambda$ weighted by the dimension of the isotypic component for $V_{\mu}\otimes V_\nu$. This defines a joint law for $\mu,\nu$ which are random partitions of length at most $m,n$ respectively, and Corollary \ref{cor: schur identity} says that this joint law is given by
\begin{equation*}
Q_{m,n,s}(\mu,\nu)=\frac{h((m^n))}{Z_{m,n,s}h(\mu)h(\nu+(s^n))h((\mu+(s^m)+(n^m))\cup\nu')}.
\end{equation*}

Note that if $s=0$, then $\lambda$ is distributed as the Poissonized Plancherel measure (with parameter $1$) conditioned to have length at most $m+n$ parts. The Poissonized Plancherel measure with parameter $\alpha$ is the measure which assigns probability $e^{-\alpha}\frac{\alpha^{|\lambda|}}{h(\lambda)^2}$ to $\lambda$ (see \cite{BOO00}, \cite{J01}).

Finally, the asymptotic behaviour of $P$ and $Q$ is computed as $m,n,s$ are sent to infinity. First, the following lemmas are proved which give asymptotics for ratios of hook lengths.

\begin{lemma}
\label{lemma: hook length ratio limit}
Let $\lambda$ be a partition. Then
\begin{equation*}
\lim _{m,n\rightarrow \infty}\frac{h((n^m))}{h(\lambda+(n^m))}=\frac{1}{h(\lambda)}\left(\alpha+1\right)^{-|\lambda|}.
\end{equation*}
\end{lemma}
where $n/m\rightarrow\alpha$ (if $\alpha=0$, then $n$ may be constant).
\begin{proof}
For $m\geq l(\lambda)$,
\begin{equation*}
\begin{split}
\frac{h((n^m))}{h(\lambda+(n^m))}&=\frac{1}{h(\lambda)}\prod _{(i,j)\in (n^m)}\frac{m+n-i-j+1}{\lambda_i+m+n-i-j+1}
\\&=\frac{1}{h(\lambda)}\prod _{i=1}^{l(\lambda)}\frac{(m+n-i)!(\lambda_i+m-i)!}{(m-i)!(\lambda_i+m+n-i)!}
\end{split}
\end{equation*}
because if $i>l(\lambda)$ then the hook lengths cancel, and this converges to $\frac{1}{h(\lambda)}(\alpha+1)^{-|\lambda|}$.
\end{proof}

\begin{lemma}
\label{lemma: hook length ratio limit 2}
Let $\mu,\nu$ be partitions. Then
\begin{equation*}
\lim _{n,m\rightarrow \infty}\frac{h((n^m))}{h((\mu+(n^m))\cup \nu')}=\frac{1}{h(\mu)h(\nu)}(\alpha+1)^{-|\mu|}(\alpha^{-1}+1)^{-|\nu|}
\end{equation*}
where the limit is taken such that $n/m=\alpha$.
\end{lemma}
\begin{proof}
The hook length of any box $(i,j)$ in $(n^m)$ with $i>l(\mu)$ and $j>l(\nu)$ is not affected by adding $\mu$ and $\nu$ to $(n^m)$. Thus, only boxes with $i\leq l(\mu)$ or $j\leq l(\nu)$ need to be considered. Break up the product into three factors depending on whether $i\leq l(\mu)$, $j\leq l(\nu)$ or both hold, obtaining
\begin{equation*}
\begin{split}
&\frac{h((n^m))}{h((\mu+(n^m))\cup \nu')}
\\=&\prod_{\substack{i\leq l(\mu)\\j\leq l(\nu)}} \frac{n+m-i-j+1}{\mu_i+\nu_j+n+m-i-j+1}\frac{h(((n-l(\nu))^{m}))}{h(\mu+((n-l(\nu))^m))}\frac{h(((m-l(\mu))^n)}{h(\nu+(m-l(\mu)))}.
\end{split}
\end{equation*}
The first factor corresponds to boxes with $i\leq l(\mu)$ and $j\leq l(\nu)$. The second and third factors correspond to boxes with $j>l(\nu)$ and $i>l(\mu)$ respectively. The boxes with $i>l(\mu)$ and $j>l(\nu)$ appear in both the second and third factors but do not actually contribute because their hook lengths in $(n^m)$ and $(\mu+(n^m))\cup \nu'$ are the same.

The first factor converges to $1$ because the number of factors is fixed, and the second and third can be computed using Lemma \ref{lemma: hook length ratio limit}, giving
\begin{equation*}
\frac{1}{h(\mu)h(\nu)}(\alpha+1)^{-|\mu|}(\alpha^{-1}+1)^{-|\nu|}.
\end{equation*}
\end{proof}

These lemmas can be used to compute the limiting distribution of $\lambda$ and $\mu,\nu$. The following proposition states that $\lambda$ converges in distribution to a Poissonized Plancherel random partition and $\mu,\nu$ converge jointly in distribution to independent Poissonized Plancherel random partitions with parameters depending on the relative rate at which $m,n,s$ go to infinity.

\begin{proposition}
Let $m,n,s\rightarrow \infty$ with $s/m\rightarrow \alpha$ and $s/n\rightarrow\beta$. Then if $\gamma=\beta/(\alpha\beta+\alpha+\beta)$ and $\delta=\alpha/(\alpha\beta+\alpha+\beta)$,
\begin{align*}
\lim _{m,n,s\rightarrow \infty}P_{m,n,s}(\lambda)&= \frac{e^{-(\gamma+\delta)}(\gamma+\delta)^{|\lambda|}}{h(\lambda)^2}
\\\lim _{m,n,s\rightarrow \infty}Q_{m,n,s}(\mu,\nu)&= \frac{e^{-(\gamma+\delta)}\gamma^{|\mu|}\delta^{|\nu|}}{h(\mu)^2h(\nu)^2},
\end{align*}
and $s$ can be taken constant by taking $\alpha,\beta\rightarrow 0$ with some specified rate.
\end{proposition}
\begin{proof}
First, renormalize by multiplying and dividing by $h((s^{m+n}))$. Then the first limit follows from Lemma \ref{lemma: hook length ratio limit} by noting that $s/(m+n)\rightarrow \frac{\alpha\beta}{\alpha+\beta}$. To see that
\begin{equation*}
\lim _{m,n,s\rightarrow \infty}\frac{1}{Z_{m,n,s}h(s^{m+n})}= e^{-\frac{\alpha+\beta}{\alpha\beta+\alpha+\beta}}
\end{equation*}
note that the summands are dominated by $h(\lambda)^{-2}$ and so the limit passes through the integral, and
\begin{equation*}
\sum _{\lambda}\frac{(\gamma+\delta)^{|\lambda|}}{h(\lambda)^2}=e^{-(\gamma+\delta)}.
\end{equation*}
The second limit can be computed by noting that by Lemma \ref{lemma: hook length ratio limit},
\begin{equation*}
\lim _{m,n,s\rightarrow \infty}\frac{h((s^n))}{h(\nu+(s^n))}= \frac{1}{h(\nu)}(\beta+1)^{-|\nu|}
\end{equation*}
and by Lemma \ref{lemma: hook length ratio limit 2},
\begin{equation*}
\begin{split}
&\lim _{m,n,s\rightarrow \infty}\frac{h((s+n)^m)}{h((\mu+(s+n)^m)\cup\nu')}
\\=& \frac{1}{h(\mu)h(\nu)}(\alpha+\alpha/\beta+1)^{-|\mu|}((\alpha+\alpha/\beta)^{-1}+1)^{-|\nu|}.
\end{split}
\end{equation*}
Finally,
\begin{equation*}
\frac{h((s^{m+n}))h((n^m))}{h((s^n))h((s+n)^m)}=1 
\end{equation*}
because $h((k^l))=\frac{C_{k+l}}{C_kC_l}$, and so
\begin{equation*}
\begin{split}
&\lim _{m,n,s\rightarrow \infty}\frac{h((m^n))}{Z_{m,n,s}h(\mu)h(\nu+(s^n))h((\mu+(s^m)+(n^m))\cup\nu')}
\\=&\lim _{m,n,s\rightarrow \infty}\frac{1}{Z_{m,n,s}h((s^{m+n}))}\left(\frac{h((s^{m+n}))h((m^n))}{h(\mu)h((s^n))h((s+n)^m)}\right)
\\&\qquad\qquad\qquad\times\left(\frac{h((s^n))}{h(\nu+(s^n))}\right)\left(\frac{h((s+n)^m)}{h((\mu+(s+n)^m)\cup \nu')}\right)
\\=&\frac{e^{-\frac{\alpha+\beta}{\alpha\beta+\alpha+\beta}}(\alpha+\alpha/\beta+1)^{-|\mu|}(\beta+\beta/\alpha+1)^{-|\nu|}}{h(\mu)^2h(\nu)^2}
\end{split}
\end{equation*}
and this gives the desired result.
\end{proof}

\bibliography{bibliography}{}
\bibliographystyle{amsplain}

\end{document}